%% file: links.tex
\documentclass[11pt]{article}
\usepackage{ifthen,amsmath,amssymb,amsthm,latexsym,color,epsfig,bm,verbatim,hyperref}
\usepackage{multirow}

\usepackage{pstricks}
\usepackage{pst-tree}
\usepackage{pst-node}
\usepackage{pst-plot}
\usepackage{wrapfig}

\def\bcr{\mathop{\rm bcr}\nolimits}

\def\itemi{\item [$(i)$]}
\def\itemii{\item [$(ii)$]}

\def\itemiii{\item [$(iii)$]}
\def\itemiv{\item [$(iv)$]}

\newcommand{\NP}{\ensuremath{\bm{\mathrm{NP}}}}
\newcommand{\coNP}{\ensuremath{\bm{\mathrm{coNP}}}}
\newcommand{\NN}{\ensuremath{\mathbb{N}}}
\newcommand{\ZN}{\ensuremath{\mathbb{Z}}}
\newcommand{\RN}{\ensuremath{\mathbb{R}}}

\newtheorem{theorem}{Theorem}[section]
\newtheorem{lemma}[theorem]{Lemma}

\newtheorem{corollary}[theorem]{Corollary}

\theoremstyle{definition}

\newtheorem{claim}{Claim}
\newtheorem{remark}[claim]{Remark}
\makeatletter
\@addtoreset{claim}{theorem}
\makeatother

\usepackage{pstricks}
\usepackage{pst-tree}
\usepackage{pst-node}

\psset{unit=0.9pt}

\begin{document}

\title{Link Crossing Number is \NP-hard}

\author{
{Arnaud de Mesmay\thanks{Partially supported by the projects ANR-16-CE40-0009-01 (GATO), ANR-18-CE40-0004-01 (FOCAL) and CNRS-PEPS-COMP3D.}
} \\
{\small Univ. Grenoble Alpes} \\[-0.03cm]
{\small CNRS, Grenoble INP, GIPSA-lab}\\[-0.03cm]
{\small 38000, Grenoble, France} \\[-0.03cm]
{\small \tt arnaud.de-mesmay@gipsa-lab.fr}
\and
{Marcus Schaefer
} \\
{\small School of Computing} \\[-0.03cm]
{\small DePaul University} \\[-0.03cm]
{\small Chicago, Illinois 60604, USA} \\[-0.03cm]
{\small \tt mschaefer@cdm.depaul.edu}
\and
{Eric Sedgwick
} \\
{\small School of Computing} \\[-0.03cm]
{\small DePaul University} \\[-0.03cm]
{\small Chicago, Illinois 60604, USA} \\[-0.03cm]
{\small \tt esedgwick@cdm.depaul.edu}\\[-0.03cm]
}

\maketitle

\begin{abstract}
 We show that determining the crossing number of a link is \NP-hard. For some weaker notions of link equivalence, we also show \NP-completeness.
\end{abstract}

\section{Introduction}

A {\em knot} is a simple closed curve in $\RN^3$. A link is a collection of (disjoint) knots in $\RN^3$, which we call {\em components} of the link. Two knots or links are called {\em (link)-equivalent} if there is an ambient isotopy between the two. From the very beginning, knots and links have been represented
in $\RN^2$ using {\em link diagrams}, projections of the three-dimensional links to $\RN^2$ which are mostly injective, except for finitely many {\em crossings} which are labeled to distinguish which arc crosses over/under the other arc at the crossing. In drawings of link diagrams,
the arc passing under is typically shown with a small gap, while the arc passing over the crossing is uninterrupted. Figure~\ref{fig:linkex} shows
crossing-minimal drawings of the unknot, a trefoil knot, which is a knot with at most three crossings which is not equivalent to the unknot, and
the Hopf link, which has two crossings, but consists of two components (so is not a knot).

\begin{figure}[htb]
  \centering
\begin{tabular}{c}
\includegraphics[height=1in]{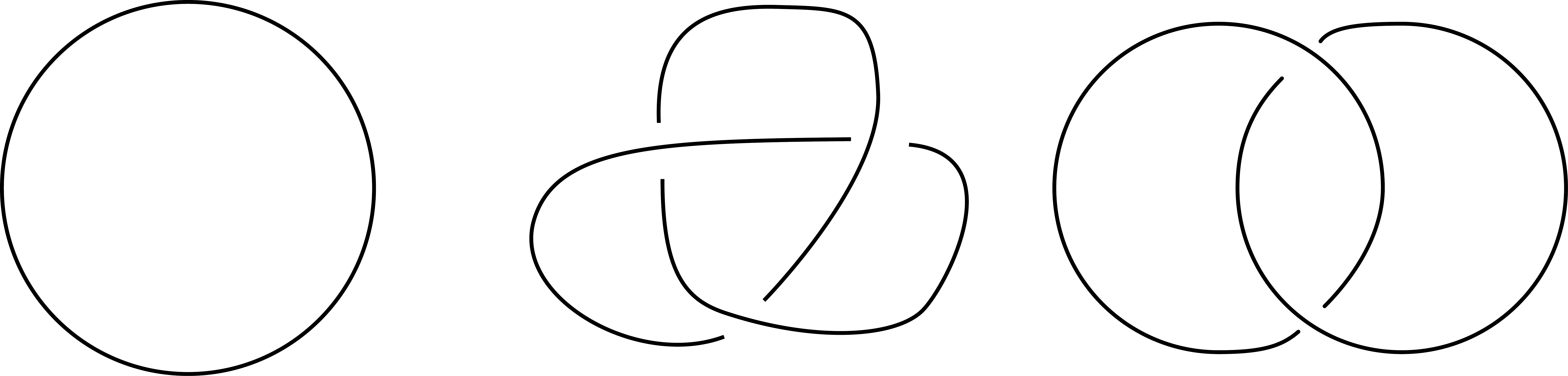}
\end{tabular}
\caption{The unknot, a trefoil knot, and the Hopf Link.}\label{fig:linkex}
\end{figure}

Two of the fundamental problems of the theory of knots and links is determining whether two diagrams represent equivalent links, and finding
a crossing-minimal drawing of a link. Both problems are hard in practice, P.G. Tait wrote in 1877 that ``it appears that the problem of finding all the absolutely distinct forms of knots with a given number of intersections is a much more difficult one than I at first thought''~\cite[pg. 315]{T1877}, and the sentiment still resonates even though many more knots and links have been classified in the meantime.

Since the initial investigations into knots and links much research has been focussed on finding improved solutions to the equivalence and crossing number problem---see, for example, Sections 2 and 5 of Lackenby's ``Elementary Knot Theory''~\cite{L17}, but as far as we know
there are no results formally showing that either of these problems is computationally hard.

In the current paper we show
that determining the link crossing number is \NP-hard. More precisely, given a link $L$ and an integer $k \in \NN$, determining
whether $c(L) \leq k$, where $c(L)$ is the smallest number of crossings in a link diagram representing a link equivalent to $L$ is
\NP-hard. Lackenby writes ``[a]lthough there seems little hope of anything more than using a brute-force search to determine the crossing number of an arbitrary link,''~\cite[Section 52]{L17}. Our main result gives a theoretical justification to this claim.

\begin{theorem}\label{thm:linkNPhard}
    Testing whether a link represented by a link diagram is equivalent to a link with at most $k$ crossings is \NP-hard.
\end{theorem}

The theorem shows that the link crossing number problem is hard, but it does not classify the complexity of the problem exactly. The main hurdle towards an upper bound for the link crossing number problem seems to be that, at least on the surface of it, it requires to check whether two link diagrams are equivalent. While it is known that the equivalence problem for link diagrams is decidable, the best computational upper bound is not elementary (the best upper bound on the number of Reidemeister moves turning one diagram into the other is a tower of exponentials of depth depending on the total number of crossings~\cite{CL16}, which is not an elementary function).

\begin{remark}
 Does the link crossing number problem remain \NP-hard if the number $k$ of crossings is fixed? For $k = 0$, testing $c(L) = 0$
 corresponds to recognizing the unknot (or unlink), as observed by Lackenby~\cite[Section 5.1]{L17}, which places the problem
 in $\NP \cap \coNP$~\cite[Section 3]{L17}, making it unlikely to be \NP-hard (since then $\NP = \coNP$ and the polynomial hierarchy collapses).
 There is no link $L$ with $c(L) = 1$. 

 More generally, if $k$ is a fixed number, then, by a result announced by Lackenby~\cite[Section 4.4]{L17}, there is a polynomial $p$ so that a link diagram with $n$ crossings is equivalent to a link diagram with at most $k$ crossings if and only if there is a sequence of at most $p(n)$ Reidemeister moves turning the diagram with $n$ crossings into a diagram with at most $k$ crossings. This implies that for any fixed $k$, the problem of testing whether $c(L) \leq k$ lies in \NP. Since $p$ depends on $k$, this is no longer true if $k$ is not bounded.
\qed\end{remark}

We can weaken the notion of link equivalence to obtain various versions of the link crossing number problem which are still \NP-hard, and some will be \NP-complete. These results will be a corollary to the proof of Theorem~\ref{thm:linkNPhard}. We use the notion of linking number
, defined in Section~\ref{sec:LW}.

We introduce the following link equivalencies: two links $L$ and $L'$ are
\begin{itemize}
 \itemi {\em parity-linking-number equivalent} if there is a bijection between the components of $L$ and $L'$ so that the parity of the linking number between any two components in $L$ is the same as the parity of the linking number between the corresponding components of $L'$ (orientation does not matter because of parity),
 \itemii {\em linking-number equivalent} if there is a bijection between the components of $L$ and $L'$ so that the linking number between any two components in $L$ is the same as the linking number between the corresponding components of $L'$ (for appropriate orientations),
 \itemiii {\em link-homotopic} if there exists a homotopy between $L$ and $L'$ during which each component of $L$ is allowed to cross itself but such that no two distinct components are allowed to intersect.
 \itemiv {\em link-concordant} if there exists an embedding $f:L \times [0,1] \rightarrow S^3 \times [0,1]$ such that $f(L \times \{0\})=L$ and $f(L \times \{1\}) = L'$.
\end{itemize}

Linking-number and parity-linking number equivalence are natural relaxations of link isotopy where one focuses on preserving the homology (respectively homology mod $2$) of the components with respect to each other. Link homotopy was introduced by Milnor~\cite{M54} as a coarser way to compare and understand links. Link concordance was first introduced for knots by Fox and Milnor~\cite{FM66} and endows links with a stronger algebraic structure than the usual notion of isotopy. It is easy to see that link homotopy and link concordance preserve linking numbers, and thus parity-linking-number equivalence is coarser than all the other equivalence relations.

\begin{corollary}\label{cor:linkNPC}
Testing whether a link represented by a link diagram is link-homotopic or concordant to a link with crossing number at most $k$ is \NP-hard. Testing whether a link represented by a link diagram has a linking-number/parity-linking number equivalent link diagram with at most $k$ crossings is \NP-complete. 
\end{corollary}

One way to interpret the corollary is that the parity-linking structure of a link is what makes the link crossing number problem \NP-hard in our proof.

The paper continues with a section on preliminaries on graph crossing number and knots, Section~\ref{sec:P}, followed by Section~\ref{sec:NPH} containing proofs of Theorem~\ref{thm:linkNPhard}, as well as Corollary~\ref{cor:linkNPC}.

\section{Preliminaries}\label{sec:P}

\subsection{Graph Crossing Number}

To show that the link crossing number problem is \NP-hard, we use an \NP-complete crossing number problem for {\em graphs}. In a {\em drawing} of a graph $G$, vertices of $G$ are placed at distinct locations in the plane, and every edge is drawn as a simple curve connecting its endpoints. We assume that edges do not pass through vertices, at most two edges are involved in any crossing, and no two edges touch (intersect without crossing).

For a bipartite graph $G = (U \cup V, E)$ with $E \subseteq U \times V$, a {\em bipartite drawing} is a drawing of $G$ in which all vertices of $U$ lie on one line, all vertices of $V$ on a parallel line, and all edges lie strictly between those two lines. For a sample bipartite graph and a drawing, with $16$ crossings, see Figure~\ref{fig:bipex}. The {\em bipartite crossing number} of $G$, $\bcr(G)$, is the smallest number of crossings in any bipartite drawing of $G$. The bipartite crossing number problem is \NP-complete, and remains so under various restrictions (see entry on bipartite crossing number in~\cite{S13}). To simplify the reduction, we work with a version of the problem in which $G$ is the disjoint union of $K_{1,4}$ graphs, $V$ consists of all degree $1$ vertices, and the order of the $V$-vertices along their line is fixed.

\begin{theorem}[Mu{\~{n}}oz,  Unger, and Vr{\v t}o~\mbox{\cite[Theorem 1]{MUV02}}]
  Determining the bipartite crossing number of a bipartite graph $G = (U \cup V, E)$ in which all vertices in $U$ have degree $4$, all vertices in $V$ have degree $1$, and the order of the $V$-vertices along their line is fixed, is \NP-complete.
\end{theorem}

\subsection{Linking and Weights}\label{sec:LW}

Two oriented knots have a {\em linking number}, a link invariant first introduced by Gauss. For a given link diagram with two oriented components $K_1$ and $K_2$, the linking number is computed as follows: trace one of the components (in the direction of its orientation); start with a count of $0$; every time you cross $K_2$: add $1$ if $K_2$ crosses under you from right to left, or $K_2$ crosses above you, from left to right. Otherwise, subtract $1$. The resulting number divided by two is the linking number of $K_1$ and $K_2$ and it is a link-invariant, that is, it is the same for all diagrams realizing the oriented link.

Suppose that a component $K$ in a link diagram has a self-crossing. We can remove the self-crossing by {\em smoothing} it, i.e., cutting the component at the self-crossing, and reconnecting the ends locally so as to avoid the crossing. There are two ways of reconnecting the ends, but only one will result in a single component $K'$, we call that one the {\em connected smoothing}. If the component is oriented, we must reverse the orientation of part of the component when reconnecting it. There are two ways of doing this, but both result in the same underlying knot, with reverse orientations. In general, $K$ and $K'$ will not be isotopic, but we do have the following lemma.

\begin{lemma}\label{lem:unknot}
   Given a link diagram realizing a link consisting of two oriented components $K$ and $K'$, the linking number of the two components does not change modulo $2$ if $(i)$ the orientation of either or both components is changed, or $(ii)$ if one does a connected smoothing.
\end{lemma}

The lemma allows us to remove self-crossings in link diagrams, as long as we are only interested in maintaining the link crossing number modulo $2$ as an invariant.

\begin{proof}
 Changing the orientation of a component does not affect the absolute value of the linking number, so $(i)$ is true in particular.
 To see $(ii)$ we perform a connected smoothing in two steps. Suppose we have $K$ and $K'$ so that $K$ has a self-crossing. Remove that self-crossing of $K$ and reconnect the ends at the self-crossing so $K$ now consists of two components $K_1$ and $K_2$, both maintaining the original orientation of $K$. Suppose $K'$ has linking number $k \in \ZN$ with $K_1$. If we reconnect $K_1$ to $K_2$ at the former self-crossing so that $K_1$ and $K_2$ form a new, single component, with the orientation along $K_1$ reversed, the linking number of $K'$ with the former $K_1$-part changes from $k$ to $-k$, while the linking number with $K_2$ remains unchanged. Therefore, the overall linking number changes by $2k$, which is zero modulo $2$.
\end{proof}

We say two components in a link are {\em unlinked} (also called {\em split}) if the link is equivalent to a link diagram in which the components are disjoint. Otherwise, the knots are {\em linked}.

Since two disjoint components have linking number $0$, we have the following result.

\begin{lemma}
 If the linking number of two components is different from $0$ modulo $2$, then the two components are {\em linked}.
\end{lemma}

It is easy to construct two linked unknots whose linking number is $0$, so the converse of the lemma is not true.

We also want to assign a linking number to two unoriented knots; to avoid confusion we call this the {\em unsigned linking number}, and we define it as the absolute value of the linking number of the two knots, for any orientation of the knots (since changing the orientation of a link only affects the sign of the linking number, this notion is well-defined).

To simplify the proof of the main result, we work with {\em weighted} links, that is, each component $K$ of the link $L$ is assigned an (integer) weight $w(K) > 0$. A crossing between a component of weight $w$ with a component of weight $w'$ is assigned a weight of $ww'$ (in the case of self-crossings, the two components are the same, so $w = w'$). The {\em weighted link crossing number} $c(D)$ of a weighted link diagram $D$ is the sum of the weights of all its crossings. For a weighted link $L^w$ given by a link diagram $D$, the {\em weighted link crossing number} is the smallest weighted link crossing number of any diagram equivalent to $D$.

While weights simplify the argument, we also need to be able to get rid of them. In graph drawing it is typically easy to replace a weighted curve by multiple, close, copies of the curve. With links we have the added complication that if a component contains a self-crossing (which it must if
it is not the unknot), replacing it with multiple copies of the same component close-by will force crossings between the components, which will
introduce crossings, throwing off the accounting. For this reason, we work with links made up of unknots drawn without self-crossings.

\begin{lemma}\label{lem:unweight}
 Given a weighted link $L^w$ with a link diagram $D$ in which all components are free of self-crossings (in particular, all components are unknots), we can (efficiently) construct a link $L'$, without weights, so that $c(L') \leq c(D)$ and there is a mapping $\alpha$ from the components of $L'$ to the components of $L^w$ so that $K$ and $K'$ in $L'$ are linked if and only if $\alpha(K)$ and $\alpha(K')$ are linked in $L^w$ and
 the linking number of $\alpha(K)$ and $\alpha(K')$ is the same as the linking number of $K$ and $K'$. Furthermore, the isotopy class of $L'$ does not depend on the link diagram $D$ and only on the weighted link $L^w$.
\end{lemma}
\begin{proof}
  By assumption, all the components of $L^w$ are drawn as crossing-free unknots in $D$. Construct $L'$ by replacing each component $K$ of $L^w$ with $w(K)$ parallel (concentric) copies. Since $K$ is drawn without crossings, we can draw each copy very close to $K$ and so that it behaves exactly the same way as $K$ with respect to all other components. The resulting drawing witnesses that $c(L') \leq c(D)$. The mapping $\alpha$ assigns to each copy in $L'$ its original in $L^w$.

  The isotopy class of $L'$ does not depend on the link diagram $D$: topologically, the construction amounts to locally replacing (cabling) each component $K$ of $L^w$ with $w(K)$ parallel copies so that they form a $w(K)$-component unlink, and there is a unique (up to isotopy) way to do so.
\end{proof}

\subsection{Framework Gadgets}

For the construction we need a device for enforcing order in a link diagram. We make use of a linked chain of unknots: an {\em $n$-chain} is a link consisting of $n$ unknots equivalent to the link pictured in Figure~\ref{fig:nchain}, for the case $n=4$, up to switching both the crossing types of successive crossings of two adjacent unknots. More precisely, an $n$-chain is a link with $n$ components that has a diagram with $2n$ crossings so that 1) no component has a self-crossing, 2) the $i$-th and $(i+1)st$ (modulo $n$) unknot have two crossings, one over and one under (in particular they have unsigned number $1$), and 3) there are no other crossings.
See Figure~\ref{fig:nchain} for a crossing-minimal drawing of a $4$-chain.

\begin{figure}[htb]
  \centering
\begin{tabular}{c}
\includegraphics[height=2in]{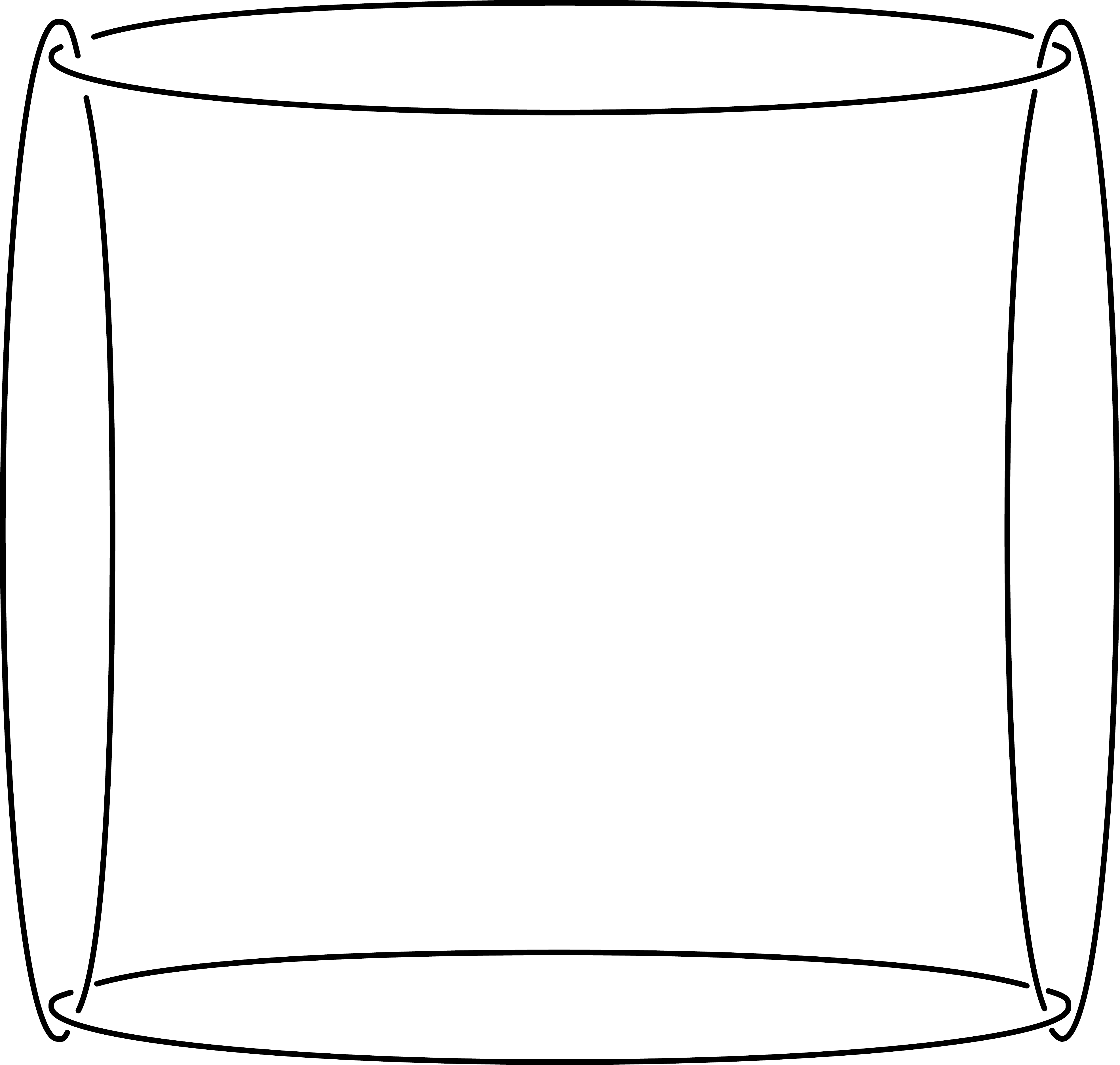}
\end{tabular}
\caption{An $n$-chain, for $n = 4$.}\label{fig:nchain}
\end{figure}

\begin{lemma}\label{lem:orderlink}
   Let $L$ be an $n$-chain, $n \geq 4$. Then $c(L) = 2n$. If $D$ is a crossing-minimal link diagram of $L$, then
   \begin{itemize}
    \itemi any two consecutive components (in the chain) cross each other exactly twice,
    \itemii there are no other crossings,
    \itemiii the two crossings a component has with another component are consecutive; that is, along each component, the four crossings with its two neighboring components do not interleave, and
    \itemiv there are two faces of $D$ each bounded by $n$ arcs, one arc from each component of $L$, occurring in the reverse order as they do along the chain. We call these the inner and outer face of $L$.
   \end{itemize}
\end{lemma}

In other words, the diagram shown in Figure~\ref{fig:nchain} is unique (up to a homeomorphism of the surface, and a change in crossing types).

\begin{proof}
    Any two consecutive components of the chain have to cross, and, since they are closed curves, have to cross at least twice,
    so the $2n$ crossings shown in Figure~\ref{fig:nchain} are unavoidable. Since we are assuming that $c(D) = c(L)$, there
    can be no additional crossings in $D$, establishing both $(i)$ and $(ii)$. In particular, $c(L) = 2n$, and there are no self-crossings in $D$.

    Since $n > 2$, each component is involved in four crossings along its boundary with two other components along the chain, say $H$ and $H'$. If the four crossings interleave, then $H$ and $H'$ must cross each other, which means, by $(ii)$, that they must be consecutive along the chain; since they both cross the same component this is not possible, since we assumed that $n \geq 4$. This establishes $(iii)$.

    The two crossings between two consecutive components form a bigon. If we remove all the arcs involved in such bigons, we are left with two disjoint sequences of arcs, both in the same order as along the chain, and both bounding a (different) face of the diagram.
\end{proof}

We make use of $n$-chains to build the framework for the reduction. Lemma~\ref{lem:chainframe} shows how we connect layers of the framework.

Suppose we have a sequence $L_1, \ldots, L_k$ of $k$ chains of length $4$, and one chain, $L_0$, of length $\ell \geq 4$. Let $L_{i,1}$, $L_{i,2}$, $L_{i,3}$, $L_{i, 4}$ be the components in chain $L_i$ for $1 \leq i\leq k$ and, for $i =0$, any four consecutive components of $L_i$. Add an unknot linking $L_{i,j}$ to $L_{i+1, j}$ for every $1 \leq i < k$, and $1 \leq j \leq 4$ in the way pictured in Figure~\ref{fig:chainframe}. We call the resulting link ${\mathcal L}(k, \ell)$ a {\em framework}, the components belonging to chains are {\em rings} and the additional unknots linking them we call {\em hinges}. See Figure~\ref{fig:chainframe} for a sample framework with $k = 3$, and $\ell = 6$.

\begin{figure}[htb]
  \centering
\begin{tabular}{c}
\includegraphics[height=3in]{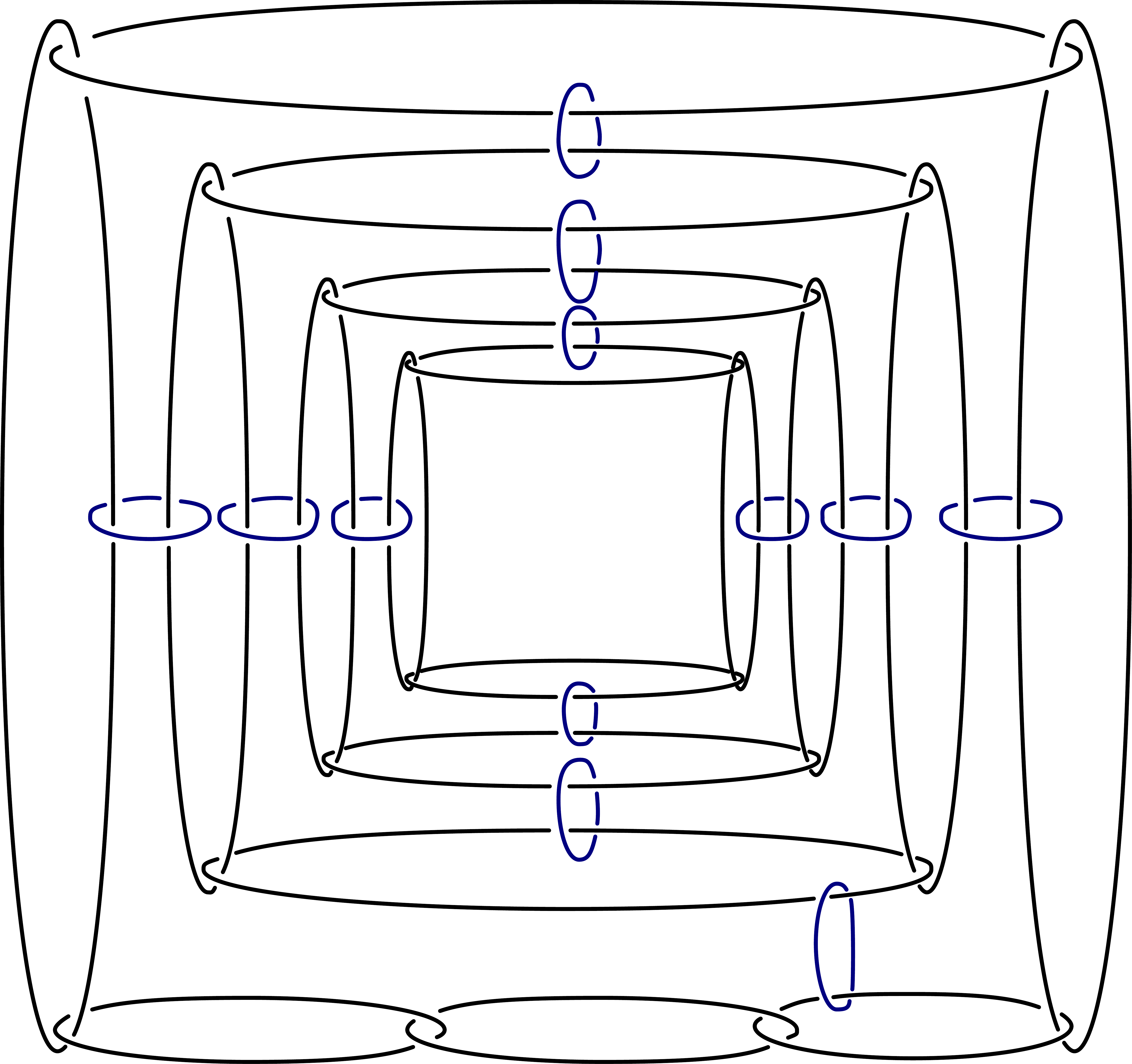}
\end{tabular}
\caption{A framework, for $k=3$ and $\ell = 6$.}\label{fig:chainframe}
\end{figure}

\begin{lemma}\label{lem:chainframe}
    Let $D$ be a crossing-minimal drawing of a framework ${\mathcal L}(k, \ell)$, with $k \geq 0$, and $\ell \geq 4$.
    Then $c(D) = c({\mathcal L}) = 24k + 2\ell$. Moreover, up to a homeomorphism of the surface we can assume that
    \begin{itemize}
        \itemi chain $L_i$ is drawn in the outer face of $L_{i+1}$, for $0 \leq i < k$,
        \itemii we can draw an (open, simple) curve $\gamma$ that crosses $L_0, \ldots, L_k$ in that order, crossing each $L_i$ exactly twice; $\gamma$ can be made to cross $L_0$ in any component not crossing a hinge if $\ell \geq 5$.
    \end{itemize}
\end{lemma}

In other words, up to a homeomorphism and changes in crossing types, we can assume that the diagram looks as pictured in Figure~\ref{fig:chainframe}.

\begin{proof}
    Figure~\ref{fig:chainframe} illustrates how to draw ${\mathcal L}(k, \ell)$. The rings contribute $8k + 2\ell$ crossings between themselves, and the hinges add $16k$ crossings to that. All of these crossings are unavoidable, so $c(D) = c({\mathcal L}) = 24k + 2\ell$. In particular, there cannot be any additional crossings (including self-crossings).

    By Lemma~\ref{lem:orderlink} that implies that each chain is drawn as described there, bounding two faces, and there is no crossing between two distinct chains. We cannot immediately assume that the two faces of each chain are inner and outer faces with regard to all chains, since the homeomorphisms for that may differ for each chain. The presence of the hinges resolves this issue, by forcing the chains to be laid out compatibly. Consider two consecutive chains, $L_i$ and $L_{i+1}$, of the framework. Apply a homeomorphism so that $L_{i+1}$ bounds an inner and an outer face, and so that $L_i$ does not lie in the inner face. Then $L_i$ must lie in the outer face of $L_{i+1}$, since otherwise the four hinges linking $L_i$ and $L_{i+1}$ cannot be drawn with at most $16$ crossings. We need to argue that $L_{i+2}$ lies in the inner face of $L_{i+1}$. Because of the hinges, it must lie either in the inner or the outer face of $L_{i+1}$, so let us consider, for a contradiction, that both $L_i$ and $L_{i+2}$ lie in the outer face of $L_{i+1}$. Contracting arcs, it is easy to see that this would result in a planar drawing of the complete bipartite graph $K_{2,4}$ in which the four vertices of degree $2$ lie in the same face; this is not possible, since we could then add a vertex to that face, and connect it to each of the four vertices, yielding a planar drawing of a $K_{3,4}$.\footnote{For this part of the proof, three hinges would be sufficient, however, we need the fourth hinge to ensure that $L_i$ does not lie inside a face bounded by arcs from three components only.}

    We conclude that the $L_i$ are drawn as concentric, disjoint chains, with $L_0$ being the outermost chain (this is arbitrary, of course, it could be innermost as well).

    We can now start in the outer face of $L_0$; if $\ell \geq 5$, there is a knot in the chain $L_k$ which does not cross a hinge. Let $\gamma$ cross through two arcs of that knot to reach the inner face of $L_0$. This is also the outer face of $L_1$. The hinges are disjoint, so $\gamma$ can be continued to cross $L_1$ and so on, until it has reached the inner face of $L_k$.
\end{proof}

\section{NP-hardness of Link Crossing Number}\label{sec:NPH}

In this section we prove Theorem~\ref{thm:linkNPhard} and Corollary~\ref{cor:linkNPC}. We start with the theorem.

Suppose we are given a bipartite graph $G = (U \cup V, E)$ in which all vertices in $U$ have degree $4$, all vertices in $V$ have degree $1$. We are also given an order of the $V$-vertices and an integer $k$. We want to efficiently construct a link diagram $D$ of a weighted link $L^w$ and a number $k'$ so that the bipartite crossing number $\bcr(G)$ of $G$ is at most $k$ if and only if
the weighted link crossing number of $L^w$ is at most $k'$ if and only if the link crossing number of the unweighted link $L'$ associated with $L$ is at most $k'$.

\begin{figure}[htb]
  \centering
    \def\svgwidth{7cm}
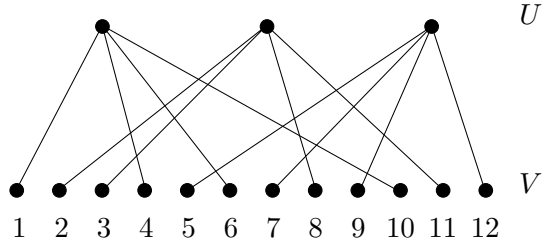
\caption{A sample bipartite graph in which all vertices of $U$ have degree $4$, all vertices of $U$ have degree $1$, and the order of the vertices in $V$ is fixed. The bipartite drawing has $16$ crossings. This is optimal for this graph (as long as the order of vertices in $V$ is unchanged).}\label{fig:bipex}
\end{figure}

We begin by constructing the framework for the reduction. Let $n_U := |U|$ and $n_V := |V| = 4n_U$.
Let $U = \{u_1, \ldots, u_{n_U}\}$, and let $v_1, \ldots, v_{n_V}$ be the vertices of $V$ in the given order.
We start with the framework ${\cal F} := {\mathcal L}(2n_U, n_V+4)$. All components in ${\cal F}$, both ring-components
and hinges, are assigned weight $w_1$ (all weights to be determined later). By Lemma~\ref{lem:chainframe},
the weighted framework ${\cal F}$ itself has crossing number $c_{\mathit{frame}} := (48n_U + 2(4n_U+4)] w_1^2 = (56n_u+8) w_1^2$.

We next add $n_U$ {\em $U$-guards} which are (crossing-free) unknots linking the $i$-th ring in ${\cal F}$ to the
$i+n_U$-th ring in ${\cal F}$ for $1 \leq i \leq n_U$. Within both rings we link the components crossed by $\gamma$ (guaranteed to
exist by Lemma~\ref{lem:chainframe}). If we think of the rings $L_i$ and $L_{i+n_U}$ as having $z$-height $i$ in the drawing,
this determines how the guards cross over/under other rings: The guard linking $L_i$ and $L_{i+n_U}$ will pass under $L_j$ for $i < j \leq n_U$ and over $L_j$ for $n_U < j \leq 2n_U$. The ring $L_0$ contains $n_V$ consecutive components which do not cross a hinge, we link a distinct unknot, called a {\em $V$-guard}, with each such component. We number the $V$-guards in the order that the chain-components they are linked to occur along
the ring $L_0$. We assign weight $w_2 < w_1$ to all guard components.

Figure~\ref{fig:linkdiag} illustrates that the framework, together with both types of guards, can be realized with $c_{\mathit{frame}} + c_{\mathit{guards}}$,
crossings, where $c_{\mathit{guards}} = (4n^2_U +2 n_V) w_1w_2$.

Finally, we link the $i$-th $U$-guard with the $j$-th $V$-guard via an unknot if there is an edge $u_iv_j$ in $G$. We call the added unknots {\em edge-links}. Edge-links have weight $1$ each. The $z$-height of the rings $L_i$, $1\leq i \leq n_U$ again determines the over/under of crossings between the edge-links as well as edge-links and rings. Any part of the edge-link belonging to edge $u_iv$ lies below any part of the edge-link belonging to $u_jv'$ for $i<j$, and edge-links cross above the rings they cross. Call the resulting weighted link $L^w$.

\begin{figure}[htb]
\begin{tabular}{c}
\includegraphics[width=0.92\textwidth]{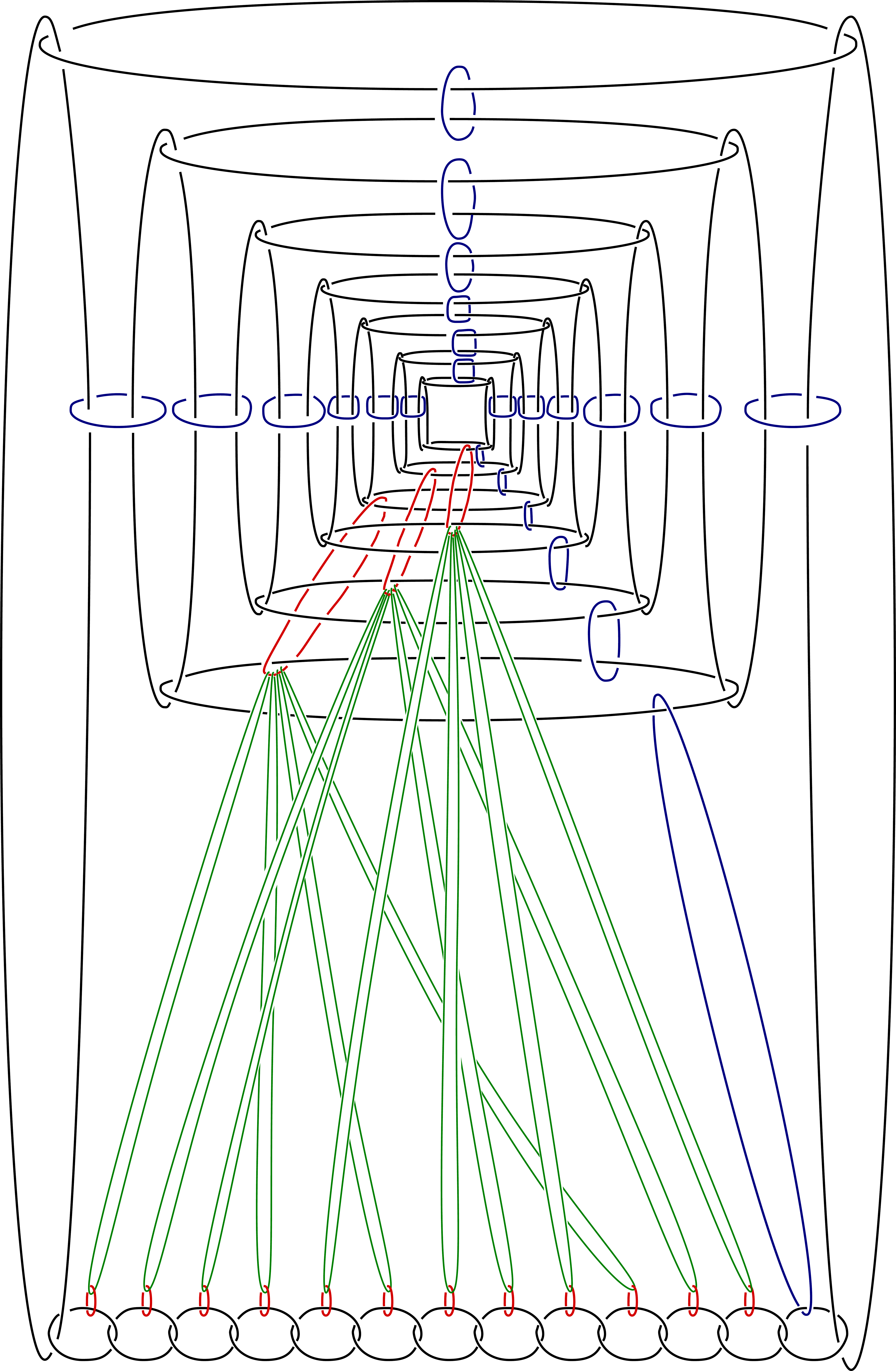}
\end{tabular}
\caption{Link diagram of $L^w$ for the sample graph $G$ from Figure~\ref{fig:bipex}. Some underedges under two adjacent subarcs of the edge links have been omitted for clarity.}\label{fig:linkdiag}
\end{figure}

Consider an arbitrary ordering of the vertices in $U$. Together with the fixed ordering $v_1, \ldots, v_{n_V}$ of $V$ this determines a bipartite drawing of $G$ with a number $c$ of crossings. We can slide the $U$-guards along the rings so their order (as seen along $L_n$, the ring that all $U$-guards have to cross) agrees with the ordering of vertices in $U$; here it matters that
we can think of the rings as having $z$-height, so $L_i$ and $L_{i+n_U}$ are
at the same $z$-height, and above later pairs of rings, that is, $L_j$ and $L_{j+n_U}$ with $1 \leq i < j \leq n_U$.
The resulting diagram of $L^w$ has a weighted link crossing number of
\[c_{\mathit{frame}} + c_{\mathit{guards}} + 4 \cdot 2 \cdot n_U^2 w1 + 4n_V w_2 + 4c,\]
where we use that $\sum_{i = 1}^{n_U} 2i-1 = n_U^2$ for counting the crossings between edge-links and rings (we also need that $\gamma$ extends to $L_0$); there are $n_V$ edge-links each resulting in four crossings with guards (two with each type); and every two edge-links that cross in $G$ result in four crossings in the diagram.

In summary: if $G$ has a bipartite drawing with crossing number at most $c$, then there is a weighted link diagram
of $L^w$ with at most $c' = c_{\mathit{frame}} + c_{\mathit{guards}} + 8 n_U^2 w1 + 4n_V w_2 + 4c$ crossings. The components
of $L^w$ are drawn without self-crossings in this diagram, so we can apply Lemma~\ref{lem:unweight} to obtain
an unweighted link $L'$, and a drawing $D$ of $L'$ with at most $c(D) \leq c'$ crossings. By construction, $L'$ does
not depend on the specific diagram of $L^w$, but only on $L^w$, and, thereby, on $G$. This completes the description of the reduction from $G$ to (unweighted) link $L'$.

We would like to argue, for the other direction, that if $L'$ has a drawing with at most
$c'$ crossings, then $G$ has a bipartite drawing with crossing number at most $c$. We cannot directly argue
that $L^w$ has a drawing with at most $c'$ weighted crossings, since the typical graph-drawing arguments do not work (it is not immediately clear, for example, that the various components in $L'$ corresponding to the same weighted component in $L^w$ can be combined so that they can be merged into a drawing of the weighted component; there may be self-crossings, and crossings between components, that make such arguments impossible).

Instead, we work with a different notion of link equivalence. Let $D$ be a drawing of a link which is
parity-link equivalent to $L'$ and having the smallest number of crossings. By assumption, that number is at most
$c'$, since the drawing of $L'$ qualifies.

By Lemma~\ref{lem:unknot} we can assume that every component in $D$ is free of self-crossings (otherwise smoothing it would decrease the number of crossings without affecting parity-linking number equivalence). Let $K$ be a weighted component in $L^w$. This component corresponds to $w(K)$ (unweighted) components in $D$. From each such group, pick the component with the fewest number of crossings, assign to it weight $w(K)$ and remove all other components belonging to the group. This yields a weighted drawing $D'$ which is parity-link-equivalent to $L^w$, and with the same weights. Moreover, the weighted crossing number of $D'$ is at most $c'$, since the crossing number did not increase.

While $D'$ will not, in general, be a diagram of $L^w$, we will now show that it is close enough to obtain a drawing of $G$ with at most $c$ crossings. Consider any two components $K$, and $K'$ of weight $w_1$ in $D'$. If the parity of the linking number is odd, then $K$ and $K'$ must cross at least twice. Since $K$ and $K'$ must also have crossed in $L^w$, this means the weight $w_1$ components contribute
at least $c_{\mathit{frame}} = (56n_u+8) w_1^2$ crossings between each other to $D'$. Without these crossings, there remain at most
$c' - c_{\mathit{frame}} = c_{\mathit{guards}} + 8 n_U^2 w1 + 4n_V w_2 + 4c$ weighted crossings. If we ensure that this number is less
than $w_1^2$, there can be no other crossings between two components of weight $w_1$. In particular, any two weight-$w_1$ components cross either twice or not at all, exactly as they did in $L^w$ (as pictured in Figure~\ref{fig:linkdiag}).

Lemma~\ref{lem:chainframe} then allows us to conclude that the frame is drawn exactly as described (and shown in Figure~\ref{fig:chainframe}).
We next turn our attention to the guards. We know there are at most $c_{\mathit{guards}} + 8 n_U^2 w1 + 4n_V w_2 + 4c$ weighted
crossing left, where $c_{\mathit{guards}} = (4n^2_U +2 n_V) w_1w_2$. If we ensure that $w_1w_2 > 8 n_U^2 w1 + 4n_V w_2 + 4c$
there can be no other crossings between rings and guards than the intended ones. That implies that rings and guards which cross, must cross two or four times: two times if they are linked, and four times if the guard passes under the ring, exactly as in the original $L^w$. In particular, the framework is drawn as shown in Figure~\ref{fig:linkdiag} (up to homeomorphism and crossing changes), for some ordering of the $U$-guards. We are left with the contribution of the edge-links;
as earlier, the crossings between edge links and rings and guards shown in Figure~\ref{fig:linkdiag} are unavoidable, so they occur in the same way, contributing $8 n_U^2 w1 + 4n_V w_2$ to the crossing number, leaving $4c$ crossings involving edge-links. Thus, if $w_1 > 4c$, there can be no additional crossings between edge links and anything in the framework, and if $w_2 > 4c$, there can be no additional crossings between edge links and guards.


We can now use the drawing $D'$
to read off a bipartite drawing of $G$: Consider $L_{n_U+1}$, and take the subarc $\ell_1$ that crosses each $U$-guard exactly twice, and nothing else. Shrink each $U$-guard to a vertex lying on the subarc $\ell_1$. Similarly, shrink the $V$-guards to vertices on a subarc $\ell_2$ of $L_0$. Since $U$-guards only cross the edge-links that are linked to them, these shrinkings do not create additional crossings between edge-links. And since edge-links do not cross $\ell_1$ nor $\ell_2$, this yields a bipartite drawing of $G$, where $\ell_1$ and $\ell_2$ are the the lines and where each edge is drawn using two arcs (the two arcs of the edge-link). Picking, for each edge, the arc with fewer crossings, we obtain a bipartite drawing of $G$ with at most $4c/4  = c$ crossings, which is what we had to show.

We still need to show that we can choose $w_1$ and $w_2$ satisfy the following four conditions:
\begin{itemize}
 \itemi $w_1^2 > c_{\mathit{guards}} + 8 n_U^2 w1 + 4n_V w_2 + 4c = (4n^2_U +2 n_V) w_1w_2 + 8 n_U^2 w1 + 4n_V w_2 + 4c$,
 \itemii $w_1w_2 > 8 n_U^2 w1 + 4n_V w_2 + 4c$,
 \itemiii $w_1 > 4c$, and
 \itemiv $w_2 > 4c$.
\end{itemize}
Using that $c \leq n_V^2$ and $n_V = 4n_U$, we can upper bound the right-hand side of $(ii)$ as
$8 n_U^2 w1 + 16n_U w_2 + 64n_U^2 < 80 n_U^2 w_1$, assuming that $w_1 > w_2$, so we can let $w_2 = 80n_U^2$ to satisfy $(ii)$.
The right-hand side of $(i)$ can be upper-bounded by $12n_U^2 w_1w_2 + 80n_U^2 w_1 < w_1 (12n_U^2 w_2 + 80 n_U^2)$, so we can
satisfy $(i)$ by letting $w_1 =  12n_U^2 w_2 + 80 n_U^2$. With these values, the four required estimates hold. This completes the proof of Theorem~\ref{thm:linkNPhard}.

For Corollary~\ref{cor:linkNPC} we observe that the \NP-hardness proof above works identically for all the equivalence notion that we have introduced. Indeed, the first part of the proof builds a diagram of $L$ with crossing number at most $c'$ from a drawing of $G$ with bipartite crossing number at most $c$, and thus this is also a valid diagram for all the notions of equivalence. The second part of the proof shows that from any parity-linking number equivalent diagram of $L$ with crossing number $c'$, one can efficiently compute a bipartite drawing of $G$ with at most $c$ crossings. Since parity-linking number equivalence refines all the other notions of equivalence, combining both sides shows that, no matter the equivalence notion, $L$ has an equivalent link diagram with at most $c'$ crossings if and only if $G$ has bipartite crossing number $\bcr(G)$ at most $c$.


For \NP-membership, it is easy to see that for parity-linking-number equivalent and linking-number equivalence, we can, in \NP, find a diagram with at most a given number of crossings, and verify equivalence efficiently by computing linking numbers.


\section{Conclusion}

The construction on Theorem~\ref{thm:linkNPhard} depends on the interaction of the various components to encode the graph crossing number problem. Can the number of components be reduced, say to a logarithmic number, or even a constant, for example by using alternating knots that are forced to have certain layouts in crossing-minimal drawings? For a single component, we have the knot crossing number problem, for which we are not aware of any upper or lower bounds (other than the upper bound based on knot equivalence).

There is a natural weakening of link equivalence which is not captured by the proof of Theorem~\ref{thm:linkNPhard}: call two links $L$ and $L'$ {\em linkedness-equivalent} if and only if there is a bijection between the components of $L$ and $L'$ so that any two components in $L$ are linked if and only if the corresponding components in $L'$ are linked. The crossing number problem for linkedness-equivalence lies in \NP, since both
linkedness and non-linkedness lie in $\NP$~\cite{HLP99,L16}, so the link crossing number problem for linkedness-equivalence lies in $\NP^{\NP \cap \coNP} = \NP$ (a basic fact from computational complexity~\cite{S83}), but our proof does not directly yield \NP-hardness of the link crossing number in this case, since the parity-linking number equivalence does not refine linkedness equivalence.

One intriguing question is based on link equivalence testing being a bottleneck in improving the upper bound on the link crossing number problem. Is testing link equivalence actually necessary? To make that precise: can the link diagram equivalence problem be reduced to the link crossing number problem? That would show that the bottleneck is essential, and cannot be removed.


\bibliographystyle{plain}
\bibliography{links}

\end{document}

%% file: graph.pdf_tex
\begingroup%
  \makeatletter%
  \providecommand\color[2][]{%
    \errmessage{(Inkscape) Color is used for the text in Inkscape, but the package 'color.sty' is not loaded}%
    \renewcommand\color[2][]{}%
  }%
  \providecommand\transparent[1]{%
    \errmessage{(Inkscape) Transparency is used (non-zero) for the text in Inkscape, but the package 'transparent.sty' is not loaded}%
    \renewcommand\transparent[1]{}%
  }%
  \providecommand\rotatebox[2]{#2}%
  \newcommand*\fsize{\dimexpr\f@size pt\relax}%
  \newcommand*\lineheight[1]{\fontsize{\fsize}{#1\fsize}\selectfont}%
  \ifx\svgwidth\undefined%
    \setlength{\unitlength}{2895.78389096bp}%
    \ifx\svgscale\undefined%
      \relax%
    \else%
      \setlength{\unitlength}{\unitlength * \real{\svgscale}}%
    \fi%
  \else%
    \setlength{\unitlength}{\svgwidth}%
  \fi%
  \global\let\svgwidth\undefined%
  \global\let\svgscale\undefined%
  \makeatother%
  \begin{picture}(1,0.41504714)%
    \lineheight{1}%
    \setlength\tabcolsep{0pt}%
    \put(0,0){\includegraphics[width=\unitlength,page=1]{graph.pdf}}%
    \put(0.99047477,0.40518835){\color[rgb]{0,0,0}\makebox(0,0)[t]{\lineheight{1.25}\smash{\begin{tabular}[t]{c}$U$\end{tabular}}}}%
    \put(0.99047477,0.08551178){\color[rgb]{0,0,0}\makebox(0,0)[t]{\lineheight{1.25}\smash{\begin{tabular}[t]{c}$V$\end{tabular}}}}%
    \put(0.0188722,0.00152262){\color[rgb]{0,0,0}\makebox(0,0)[t]{\lineheight{1.25}\smash{\begin{tabular}[t]{c}$1$\end{tabular}}}}%
    \put(0.09932541,0.00152262){\color[rgb]{0,0,0}\makebox(0,0)[t]{\lineheight{1.25}\smash{\begin{tabular}[t]{c}$2$\end{tabular}}}}%
    \put(0.17977851,0.00152262){\color[rgb]{0,0,0}\makebox(0,0)[t]{\lineheight{1.25}\smash{\begin{tabular}[t]{c}$3$\end{tabular}}}}%
    \put(0.26023167,0.00152262){\color[rgb]{0,0,0}\makebox(0,0)[t]{\lineheight{1.25}\smash{\begin{tabular}[t]{c}$4$\end{tabular}}}}%
    \put(0.34068478,0.00152262){\color[rgb]{0,0,0}\makebox(0,0)[t]{\lineheight{1.25}\smash{\begin{tabular}[t]{c}$5$\end{tabular}}}}%
    \put(0.42113791,0.00152262){\color[rgb]{0,0,0}\makebox(0,0)[t]{\lineheight{1.25}\smash{\begin{tabular}[t]{c}$6$\end{tabular}}}}%
    \put(0.50159106,0.00152262){\color[rgb]{0,0,0}\makebox(0,0)[t]{\lineheight{1.25}\smash{\begin{tabular}[t]{c}$7$\end{tabular}}}}%
    \put(0.58204418,0.00152262){\color[rgb]{0,0,0}\makebox(0,0)[t]{\lineheight{1.25}\smash{\begin{tabular}[t]{c}$8$\end{tabular}}}}%
    \put(0.66249736,0.00152262){\color[rgb]{0,0,0}\makebox(0,0)[t]{\lineheight{1.25}\smash{\begin{tabular}[t]{c}$9$\end{tabular}}}}%
    \put(0.7429505,0.00152262){\color[rgb]{0,0,0}\makebox(0,0)[t]{\lineheight{1.25}\smash{\begin{tabular}[t]{c}$10$\end{tabular}}}}%
    \put(0.82340367,0.00152262){\color[rgb]{0,0,0}\makebox(0,0)[t]{\lineheight{1.25}\smash{\begin{tabular}[t]{c}$11$\end{tabular}}}}%
    \put(0.90385677,0.00152262){\color[rgb]{0,0,0}\makebox(0,0)[t]{\lineheight{1.25}\smash{\begin{tabular}[t]{c}$12$\end{tabular}}}}%
  \end{picture}%
\endgroup%